\documentclass[a4paper,10pt]{article}
\usepackage{test,paralist}
\usepackage[authoryear,round]{natbib}
\usepackage{color}
\usepackage[colorlinks,citecolor=blue,urlcolor=magenta]{hyperref}
\usepackage{doi}

\title{Necessity of Hyperbolic Absolute Risk Aversion for the Concavity of Consumption Functions}
\author{Alexis Akira Toda\thanks{Department of Economics, University of California San Diego. Email: \href{mailto:atoda@ucsd.edu}{atoda@ucsd.edu}.}}

\begin{document}

\maketitle

\begin{abstract}

\cite{CarrollKimball1996} have shown that, in the class of utility functions that are strictly increasing, strictly concave, and have nonnegative third derivatives, hyperbolic absolute risk aversion (HARA) is sufficient for the concavity of consumption functions in general consumption-saving problems. This paper shows that HARA is necessary, implying the concavity of consumption is not a robust prediction outside the HARA class.

\medskip


\medskip

\end{abstract}

\section{Introduction}

The notion that the marginal propensity to consume decreases with wealth, or that the consumption function is concave, dates back at least to \cite{keynes1936}. This observation is important in macroeconomics because the effect of a fiscal transfer of one dollar to a wealthy household is smaller than that to a poor household, implying that fiscal policies need to account for household heterogeneity. In an important contribution, \cite{CarrollKimball1996} have shown that in the class of utility functions that are strictly increasing, are strictly concave, and have nonnegative third derivatives, hyperbolic absolute risk aversion (HARA) is sufficient for the concavity of consumption functions in finite-horizon consumption-saving problems without liquidity constraints. Their result has been extended in several directions: \citet[Section 5]{CarrollKimball2001} obtain concavity under finite-horizon, HARA utility, and liquidity constraints; \cite{NishiyamaKato2012} obtain concavity under infinite-horizon, quadratic utility, and liquidity constraints; \citet[Proposition 2.5 and Remark 2.1]{MaStachurskiToda2020JET} obtain concavity under infinite horizon, constant relative risk aversion (CRRA) utility, and liquidity constraints.\footnote{\label{fn:MST}It is clear from the proof technique of \cite{MaStachurskiToda2020JET} that concavity obtains in finite- or infinite-horizon, with or without liquidity constraints, and HARA utility if the domain of the utility function is appropriately modified. See Section \ref{sec:discussion} for more discussion.}

In all of these theoretical papers, the utility function is restricted to be HARA. Thus a natural question is whether it is possible to obtain the concavity of consumption functions under weaker assumptions: does concavity hold in a larger class than HARA, or is concavity a non-robust prediction that fails to hold outside the HARA class? This paper provides a definitive answer to this question, which is the latter. More precisely, I show that if the utility function is strictly increasing, is strictly concave, has a positive third derivative, but is not HARA, then there exists a finite-horizon (in fact, one period) consumption-saving problem such that the consumption function is not concave. Combined with the earlier results on the concavity of consumption functions just cited, my result shows that in the class of natural utility functions, HARA is both necessary and sufficient for the concavity of consumption functions in general consumption-saving problems.

\section{Main result}\label{sec:main}

Following \cite{CarrollKimball1996}, I consider a general consumption-saving problem with a finite horizon. The problem can be informally described as follows. The agent is endowed with initial wealth $w_0>0$ and lives for $T$ periods. The period utility function is $u:(0,\infty)\to \R$, which is strictly increasing and concave. The agent receives income $Y_t\ge 0$ at the beginning of time $t$. The gross return on wealth between time $t-1$ and $t$ is $R_t>0$. The agent discounts utility between time $t-1$ and $t$ using the discount factor $\beta_t>0$, where $\beta_0\equiv 1$. Letting $c_t>0$ be the consumption at time $t$ and $w_t\ge 0$ the financial wealth at the beginning of time $t$, the agent's objective is to solve
\begin{subequations}\label{eq:IF}
\begin{align}
&\maximize && \E_0\sum_{t=0}^T\left(\prod_{s=0}^t\beta_s\right)u(c_t)\label{eq:IF_utility}\\
&\st && w_{t+1}=R_{t+1}(w_t-c_t)+Y_{t+1}\ge 0,\label{eq:IF_budget}
\end{align}
\end{subequations}
where the initial wealth $w_0=w>0$ is given and \eqref{eq:IF_budget} is the budget and borrowing constraint. The discount factor $\beta_t$, return on wealth $R_t$, and income $Y_t$ can all be stochastic.

If a solution to the consumption-saving problem \eqref{eq:IF} exists, the time $t$ consumption $c_t$ can be viewed as a function of the current wealth $w_t$, which we call the consumption function. The main result of this paper is that if the utility function satisfies some regularity conditions and the consumption function is always concave regardless of the specification of the process $(\beta_t,R_t,Y_t)_{t=1}^T$, then $u$ must exhibit hyperbolic absolute risk aversion (HARA).

We now make this statement more precise. If the consumption function is concave regardless of the specification of the process $(\beta_t,R_t,Y_t)_{t=1}^T$, then in particular the consumption function in a one period problem ($T=1$) must always be concave. We can then rewrite \eqref{eq:IF} as the static maximization problem
\begin{equation}
\max_c u(c)+\E[\beta u(R(w-c)+Y)], \label{eq:IFstat}
\end{equation}
where $\beta,R,Y>0$ are arbitrary random variables and $c$ satisfies $R(w-c)+Y\ge 0$. Furthermore, if $\beta,R,Y$ are arbitrary, in particular we may restrict attention to problems in which $\beta,R,Y$ take finitely many values. In this case the expectation in \eqref{eq:IFstat} is a finite sum and always well-defined.

The following lemma shows that under standard monotonicity, concavity, and Inada conditions, the consumption and saving functions can be unambiguously defined, which are differentiable and strictly increasing.

\begin{lem}\label{lem:cfunc}
Suppose that
\begin{inparaenum}[(i)]
\item the utility function $u:[0,\infty)\to \R\cup \set{-\infty}$ is twice continuously differentiable on $(0,\infty)$ and satisfies $u'>0$, $u''<0$, and $\lim_{x\downarrow 0}u'(x)=\infty$, and
\item the random vector $(\beta,R,Y)\gg 0$ has finite support.
\end{inparaenum}
Then for any initial wealth $w>0$, the consumption-saving problem \eqref{eq:IFstat} has a unique solution, and the consumption function $c=c(w)$ and saving function $s(w)=w-c(w)$ satisfy the Euler equation
\begin{equation}
u'(c(w))=\E[\beta Ru'(Rs(w)+Y)].\label{eq:euler}
\end{equation}
Furthermore, $c,s$ are continuously differentiable and $c'(w),s'(w)\in (0,1)$.
\end{lem}

In what follows, we maintain the assumptions of Lemma \ref{lem:cfunc}. Since the saving function $s$ is continuously differentiable and strictly increasing, its range $S\coloneqq s((0,\infty))$ is an open interval of $\R$. Since $u'$ is strictly decreasing, by \eqref{eq:euler} for each $s\in S$ there exists a unique $c\in (0,\infty)$ such that $u'(c)=\E[\beta Ru'(Rs+Y)]$. Therefore we can unambiguously define the function $g:S\to (0,\infty)$ by
\begin{equation}
g(s)\coloneqq (u')^{-1}(\E[\beta Ru'(Rs+Y)]).\label{eq:defg}
\end{equation}
The function $g$ in \eqref{eq:defg} returns the consumption level $c>0$ that implies the saving $s\in S$. Since $u'$ is strictly decreasing, $g$ is clearly strictly increasing. \citet[Proposition 2.5]{MaStachurskiToda2020JET} show that the concavity of $g$ is sufficient for the concavity of the consumption function. The following lemma shows that the concavity of $g$ is actually necessary.

\begin{lem}\label{lem:gconcave}
If $c$ is concave, then so is $g$ in \eqref{eq:defg}.
\end{lem}

The following lemma, which is closely related to \citet[Section 3.16]{HardyLittlewoodPolyaInequalities}, plays a crucial role in the proof of the main result.

\begin{lem}\label{lem:HLP}
Let $I\subset \R$ be an open interval and $\phi:I \to \R$ be a twice differentiable function such that $\phi(I)=(0,\infty)$, $\phi'<0$, and either $\phi''>0$ on $I$ or $\phi''\equiv 0$ on $I$. For $p,x,v\in \R^N$ with $p,v\gg 0$ and $x\in I^N$, let
\begin{equation*}
g(s;p,x,v)\coloneqq \phi^{-1}\left(\sum_{n=1}^N p_n\phi(x_n+v_ns)\right),
\end{equation*}
which is well-defined in the neighborhood of $s=0$. Then the following statements are true:
\begin{enumerate}
\item If $g(s;p,x,v)$ is concave in the neighborhood of $s=0$ for arbitrary $p,x,v$, then $\phi$ and $I$ take one of the following forms:
\begin{subequations}\label{eq:phiI}
\begin{align}
&\phi(x)=c(ax+b)^{-1/a}, && 1<a<0, &&I=(-\infty,-b/a), \label{eq:phiI.a}\\
&\phi(x)=-bc\e^{-x/b}, &&b>0, &&I=\R,~\text{or} \label{eq:phiI.b}\\
&\phi(x)=c(ax+b)^{-1/a}, &&a>0, &&I=(-b/a,\infty), \label{eq:phiI.c}
\end{align}
\end{subequations}
where $c>0$ is arbitrary.
\item Conversely, if $\phi$ and $I$ take one of the forms in \eqref{eq:phiI}, then $g(s;p,x,v)$ is concave in $s$ on its domain for arbitrary $p,x,v$.
\end{enumerate}
\end{lem}

Note that in either case in Lemma \ref{lem:HLP}, log-differentiating $\phi$, we obtain
\begin{equation}
\frac{\phi'(x)}{\phi(x)}=-\frac{1}{ax+b},\label{eq:HARA}
\end{equation}
where $a>-1$ with $I=\set{x\in \R:ax+b>0}$. We can now state the main result.

\begin{thm}\label{thm:cconcave}
Suppose that
\begin{inparaenum}[(i)]
\item the utility function $u:[0,\infty)\to \R\cup \set{-\infty}$ is three times differentiable on $(0,\infty)$ and satisfies $u'>0$, $u''<0$, $u'''>0$, $\lim_{x\downarrow 0}u'(x)=\infty$, and $\lim_{x\to\infty}u'(x)=0$, and
\item the random vector $(\beta,R,Y)\gg 0$ has finite support.
\end{inparaenum}
If the consumption function in Lemma \ref{lem:cfunc} is concave for arbitrary distribution of $(\beta,R,Y)$, then $u$ exhibits constant relative risk aversion (CRRA). Conversely, if $u$ is CRRA, then the consumption function of the consumption-saving problem \eqref{eq:IF} is concave.
\end{thm}

\begin{proof}
Let $c,s$ be the consumption and saving functions. By Lemma \ref{lem:cfunc}, $c,s$ are continuously differentiable, strictly increasing, and $S=s((0,\infty))$ is an open interval of $\R$. Let us show $0\in S$. If $0\notin S$, then either $S\subset (0,\infty)$ or $S\subset (-\infty,0)$.

If $S\subset (0,\infty)$, then $s(w)>0$ for all $w$. Since $c(w)+s(w)=w$, we have $c(w)<w$. Using the Euler equation \eqref{eq:euler} and $u''<0$, we obtain
\begin{equation*}
u'(w)<u'(c(w))=\E[\beta Ru'(Rs(w)+Y)]\le \E[\beta Ru'(Y)]<\infty.
\end{equation*}
Letting $w\downarrow 0$, we obtain a contradiction because $u'(0)=\infty$.

If $S\subset (-\infty,0)$, then $s(w)<0$ for all $w$. Since $c(w)+s(w)=w$, we have $c(w)>w$. Again by \eqref{eq:euler} and $u''<0$, we obtain
\begin{equation*}
u'(w)>u'(c(w))=\E[\beta Ru'(Rs(w)+Y)]\ge \E[\beta Ru'(Y)]>0.
\end{equation*}
Letting $w\uparrow \infty$, we obtain a contradiction because $u'(\infty)=0$.

The above argument shows that $0\in S$ regardless of the specification of $(\beta,R,Y)$. Since $(\beta,R,Y)$ has finite support, we can write $g(s)$ in \eqref{eq:defg} as
\begin{equation}
g(s)=(u')^{-1}\left(\sum_{n=1}^N\pi_n\beta_nR_nu'(R_ns+Y_n)\right),\label{eq:gs}
\end{equation}
where $\pi_n>0$ is the probability of state $n$. Since $(\beta_n,R_n,Y_n)\gg 0$ is arbitrary, we can apply Lemma \ref{lem:HLP} for $\phi=u'$, $I=(0,\infty)$, $p_n=\pi_n\beta_nR_n$, $x_n=Y_n$, and $v_n=R_n$, yielding $u'(x)=\phi(x)=c(ax+b)^{-1/a}$ for some $a,c>0$ and $b=0$. Since the multiplicative constant $c$ does not affect the ordering of utility, we may assume $u'(x)=x^{-\gamma}$ with $\gamma=1/a>0$, so $u$ is CRRA.

The sufficiency of CRRA for the concavity of consumption can be shown by the same argument as in \citet[Proposition 2.5 and Remark 2.1]{MaStachurskiToda2020JET}.
\end{proof}

One limitation of Theorem \ref{thm:cconcave} is that to derive the strong conclusion that the utility function is CRRA, we require the strong assumption that the distribution of $(\beta,R,Y)$ is arbitrary. The following corollary shows that we obtain the same conclusion even if the discount factor $\beta$ is exogenously fixed at a constant value.

\begin{cor}\label{cor:exobeta}
Suppose that the assumptions of Theorem \ref{thm:cconcave} hold and the discount factor $\beta>0$ is exogenously given. If the consumption function in Lemma \ref{lem:cfunc} is concave for arbitrary distribution of $(R,Y)$, then $u$ is CRRA.
\end{cor}

\begin{proof}
If the consumption function $c$ is concave, by Lemma \ref{lem:gconcave} and the proof of Theorem \ref{thm:cconcave}, $g$ in \eqref{eq:gs} is concave in the neighborhood of $0\in S$. Therefore setting $\beta_n=\beta$ and $s=ks$ for any constant $\beta,k>0$, the function
\begin{align}
h(s)\coloneqq g(ks)&=(u')^{-1}\left(\sum_{n=1}^N\pi_n\beta R_nu'(R_nks+Y_n)\right) \notag \\
&=(u')^{-1}\left(\sum_{n=1}^N p_nu'(x_n+v_ns)\right) \label{eq:defh}
\end{align}
is concave in the neighborhood of $0\in S$, where
\begin{equation}
p_n=\pi_n \beta R_n, \quad x_n=Y_n, \quad \text{and} \quad v_n=R_nk. \label{eq:pxYn}
\end{equation}
Since $\frac{p_n}{v_n}=\frac{\beta \pi_n}{k}$ and $\sum_{n=1}^N \pi_n=1$, it must be
\begin{equation}
k=\frac{\beta}{\sum_{n=1}^N p_n/v_n}.\label{eq:defk}
\end{equation}
Therefore given any $p,x,v\gg 0$, we can choose $R,Y\gg 0$ and distribution $\pi$ such that \eqref{eq:pxYn} holds by choosing $k$ as in \eqref{eq:defk} and setting $R_n=v_n/k$, $Y_n=x_n$, and $\pi_n=\frac{kp_n}{\beta v_n}$. This argument shows that if the consumption function is concave for arbitrary distribution of $(R,Y)$, then the function $h$ in \eqref{eq:defh} is concave in the neighborhood of $s=0$ for arbitrary $p,x,v\gg 0$. The conclusion then follows from Lemma \ref{lem:HLP} and the proof of Theorem \ref{thm:cconcave}.
\end{proof}

\section{Discussion}\label{sec:discussion}

Theorem \ref{thm:cconcave} essentially shows that for the consumption functions in general consumption-saving problems to be concave, constant relative risk aversion (CRRA) is necessary and sufficient. This statement may appear to be in conflict with \cite{CarrollKimball1996}, who have shown that hyperbolic absolute risk aversion (HARA) is sufficient for concavity. However, there is no contradiction between the two results. This is because in Theorem \ref{thm:cconcave}, to avoid unnecessary complications, I have restricted the utility function to have domain $(0,\infty)$ and satisfy appropriate Inada conditions. If the domain is changed to $(-b/a,\infty)$, then it is straightforward to adopt the proof of Theorem \ref{thm:cconcave} to show that $u$ is HARA, as we can see from \eqref{eq:HARA} with $\phi=u'$. This argument shows that for the concavity of consumption functions, HARA is necessary and sufficient (among the class of utility functions with $u'>0$, $u''<0$, and $u'''\ge 0$).

Next, we discuss some subtleties regarding the concavity of consumption functions in the existing literature.

When describing the consumption-saving problem, \cite{CarrollKimball1996} are unclear about the domain of the utility function, Inada conditions, and borrowing constraints. Since the proof of their Lemma 2 uses first-order conditions (that hold for interior solutions), they implicitly assume that consumption can be any value in the domain of the HARA utility (which trivially satisfies the Inada conditions) and that the agent can save or borrow freely. In the consumption-saving problems \eqref{eq:IF}, \eqref{eq:IFstat}, and Lemma \ref{lem:cfunc}, these assumptions were made explicit. However, note that the possibility of borrowing is used in Theorem \ref{thm:cconcave} to show the \emph{necessity} of HARA for the concavity of consumption functions; the possibility of borrowing is not required for the \emph{sufficiency} of HARA for concavity, and in fact \citet[Remark 2.1]{MaStachurskiToda2020JET} show that HARA is sufficient even in the presence of borrowing constraints.\footnote{Strictly speaking, \citet[Remark 2.1]{MaStachurskiToda2020JET} concerns CRRA utility, but this is because the domain of the utility function is restricted to $(0,\infty)$. It is straightforward to handle HARA utility by shifting the domain.}

Theorem \ref{thm:cconcave} and the subsequent discussion suggest that we cannot expect the concavity of consumption functions unless the stochastic process $(\beta_t,R_t,Y_t)_{t=1}^T$ is somehow restricted. Indeed, \cite{GongZhongZou2012} show that in a finite-horizon, deterministic consumption-saving problem in which the discount factor $\beta_t$ and the gross risk-free rate $R_t$ satisfy $\beta_tR_t\ge 1$, the concavity of the absolute risk tolerance $-u'(x)/u''(x)$ is sufficient for the concavity of consumption functions. This condition is much weaker than HARA, as the absolute risk tolerance is affine when $u$ is HARA. However, note that the condition $\beta R\ge 1$ is quite restrictive, and in fact not satisfied in stationary general equilibrium models with infinitely-lived agents (see the discussion in the proof of Theorem 8 of \citealp{StachurskiToda2019JET}). Corollary \ref{cor:exobeta} shows that the possibility of stochastic discounting plays no role in the conclusion.

Finally, in the more recent literature, several authors have studied consumption-saving problems with hyperbolic discounting \citep{CaoWerning2018,MorrisPostlewaite2020}. Although the concavity of consumption functions has not been discussed in this context, since exponential discounting is a special case of hyperbolic discounting, the result in this paper immediately implies that we cannot expect concavity unless the utility function is HARA.

\appendix

\section{Proof of lemmas}

\begin{proof}[Proof of Lemma \ref{lem:cfunc}]
Fix $w>0$. Suppose $(\beta,R,Y)$ takes the value $(\beta_n,R_n,Y_n)$ with probability $\pi_n>0$, where $n=1,\dots,N$. Define
\begin{equation*}
\bar{c}=\sup\set{c>0: (\forall n) R_n(w-c)+Y_n\ge 0}=w+\min_n \frac{Y_n}{R_n}\ge w>0,
\end{equation*}
which is well-defined because $R_n,Y_n>0$. Define $f:(0,\bar{c})\to \R$ 
as the objective function in \eqref{eq:IFstat}. Since
\begin{align*}
f'(c)&=u'(c)-\E[\beta Ru'(R(w-c)+Y)],\\
f''(c)&=u''(c)+\E[\beta R^2u''(R(w-c)+Y)]<0
\end{align*}
because $u''<0$, $f$ is strictly concave. By the Inada condition, we have $\lim_{c\downarrow 0}f'(c)=\infty$ and $\lim_{c\uparrow \bar{c}}f'(c)=-\infty$. Since $f'$ is continuous and strictly decreasing, by the intermediate value theorem there exists a unique $c^*\in (0,\bar{c})$ satisfying $f'(c^*)=0$. By the strict concavity of $f$, this $c^*\eqqcolon c(w)$ uniquely solves the optimization problem \eqref{eq:IFstat}. Letting $s(w)=w-c(w)$, the first-order condition $f'(c(w))=0$ implies the Euler equation \eqref{eq:euler}.

To show the continuous differentiability of $c$, let
\begin{equation*}
F(w,c)=u'(c)-\E[\beta Ru'(R(w-c)+Y)].
\end{equation*}
Since $u''<0$, we obtain
\begin{align*}
\frac{\partial F}{\partial w}&=-\E[\beta R^2u''(R(w-c)+Y)]>0,\\
\frac{\partial F}{\partial c}&=u''(c)+\E[\beta R^2u''(R(w-c)+Y)]<0.
\end{align*}
By the implicit function theorem, $c$ is continuously differentiable and
\begin{equation*}
c'(w)=-\frac{\partial F/\partial w}{\partial F/\partial c}=\frac{-\E[\beta R^2u''(R(w-c)+Y)]}{-u''(c)-\E[\beta R^2u''(R(w-c)+Y)]}\in (0,1).
\end{equation*}
Then $s'(w)=1-c'(w)\in (0,1)$.
\end{proof}

\begin{proof}[Proof of Lemma \ref{lem:gconcave}]
By the discussion preceding Lemma \ref{lem:gconcave}, we have
\begin{equation}
g(s(w))=c(w)\label{eq:gsw}
\end{equation}
for all $w>0$. If $g$ is not concave, we can take $s_1,s_2\in S$ and $\alpha\in [0,1]$ such that
\begin{equation}
g((1-\alpha)s_1+\alpha s_2)<(1-\alpha)g(s_1)+\alpha g(s_2).\label{eq:gnonconcave}
\end{equation}
Let $w_j=s^{-1}(s_j)$ for $j=1,2$. Then $s_j=s(w_j)$ and $g(s_j)=c(w_j)$ by \eqref{eq:gsw}. Define $\bar{c}\coloneqq (1-\alpha)c(w_1)+\alpha c(w_2)$ and $\bar{w}\coloneqq (1-\alpha)w_1+\alpha w_2$. Since by assumption $c$ is concave, we have
\begin{equation}
c(\bar{w})=c((1-\alpha)w_1+\alpha w_2)\ge (1-\alpha)c(w_1)+\alpha c(w_2)=\bar{c}.\label{eq:cineq}
\end{equation}
Therefore
\begin{align*}
c(\bar{w})&\ge \bar{c}=(1-\alpha)c(w_1)+\alpha c(w_2) && (\because \eqref{eq:cineq}) \\
&=(1-\alpha)g(s_1)+\alpha g(s_2)>g((1-\alpha)s_1+\alpha s_2) && (\because \eqref{eq:gsw}, \eqref{eq:gnonconcave}) \\
&=g((1-\alpha)(w_1-c(w_1))+\alpha(w_2-c(w_2)))\\
&=g(\bar{w}-\bar{c})\ge g(\bar{w}-c(\bar{w})) && (\because \eqref{eq:cineq}, \text{$g$ increasing})\\
&=g(s(\bar{w}))=c(\bar{w}), && (\because \eqref{eq:gsw})
\end{align*}
which is a contradiction.
\end{proof}

\begin{proof}[Proof of Lemma \ref{lem:HLP}]
The proof is similar to \citet[Section 3.16]{HardyLittlewoodPolyaInequalities}, who assume $\phi'>0$ and the goal is to characterize the \emph{convexity} of $g$. Since these authors do not cover all the fine details, I reproduce their argument.

To simplify the notation, write $g=g(s;p,x,v)$ and $\sum=\sum_{n=1}^N$. Since $x_n\in I$ for all $I$, $I$ is open, $\phi(I)=(0,\infty)$, and $\phi$ is monotonic, $g$ is well-defined in a neighborhood of 0. Differentiating both sides of $\phi(g(s))=\sum \phi(x_n+v_ns)$ with respect to $s$ twice, we obtain
\begin{align*}
\phi'(g)g'&=\sum p_nv_n\phi'(x_n+v_ns),\\
\phi''(g)(g')^2+\phi'(g)g''&=\sum p_nv_n^2\phi''(x_n+v_ns).
\end{align*}
Eliminating $g'$, we obtain
\begin{equation*}
\phi'(g)^3g''=\phi'(g)^2\sum p_nv_n^2\phi''(x_n+v_ns)-\phi''(g)\left(\sum p_nv_n\phi'(x_n+v_ns)\right)^2.
\end{equation*}
Since by assumption $\phi'<0$, we have $g''(s)\le 0$ if and only if
\begin{equation}
\phi'(g)^2\sum p_nv_n^2\phi''(x_n+v_ns)-\phi''(g)\left(\sum p_nv_n\phi'(x_n+v_ns)\right)^2\ge 0.\label{eq:gconcave}
\end{equation}
If $\phi''\equiv 0$ on $I$, \eqref{eq:gconcave} is trivial. In this case $\phi$ is quadratic and (because $\phi>0$, $\phi'<0$, and $\phi(I)=(0,\infty)$) we must have $\phi(x)=c(ax+b)^{-1/a}$ for $a=-1/2$ and $c>0$.

Therefore it remains to consider the case $\phi''>0$ on $I$. If \eqref{eq:gconcave} holds for all $s$ in a neighborhood of 0 for arbitrary $x\in I^N$ and $p,v\gg 0$, in particular letting $s=0$, we obtain
\begin{align}
&\phi'(g)^2\sum p_nv_n^2\phi''(x_n)-\phi''(g)\left(\sum p_nv_n\phi'(x_n)\right)^2\ge 0\notag \\
\iff & \frac{\phi'(g)^2}{\phi''(g)}\ge \frac{\left(\sum p_nv_n\phi'_n\right)^2}{\sum p_nv_n^2\phi''_n},\label{eq:lb}
\end{align}
where  $\phi_n=\phi(x_n)$ and $\phi_n',\phi_n''$ are defined analogously. 
Applying the Cauchy-Schwarz inequality, we obtain
\begin{equation*}
\left(\sum p_nv_n\phi'_n\right)^2=\left(\sum \sqrt{p_n\phi_n''}v_n\sqrt{\frac{p_n{\phi'_n}^2}{\phi_n''}}\right)^2\le \left(\sum p_nv_n^2\phi_n''\right)\left(\sum \frac{p_n{\phi_n'}^2}{\phi_n''}\right),
\end{equation*}
with equality achieved when $v_n=k\phi_n'/\phi_n''$ for some $k<0$ (so that $v_n>0$). Therefore \eqref{eq:lb} for all $x\in I^N$ and $p,v\gg 0$ is equivalent to
\begin{equation}
\frac{\phi'(g)^2}{\phi''(g)}\ge \sum \frac{p_n{\phi_n'}^2}{\phi_n''}\quad \text{for all $x\in I^N$ and $p\gg 0$}.\label{eq:lb2}
\end{equation}
Now define $y_n=\phi(x_n)$ and 
\begin{equation}
\Phi(y)\coloneqq \frac{[\phi'(\phi^{-1}(y))]^2}{\phi''(\phi^{-1}(y))}.\label{eq:defPhi}
\end{equation}
Noting that $\phi(I)=(0,\infty)$, \eqref{eq:lb2} is equivalent to
\begin{equation}
\Phi\left(\sum p_ny_n\right)\ge \sum p_n\Phi(y_n) \quad \text{for all $p,y\in \R_{++}^N$}.\label{eq:lb3}
\end{equation}
If we take $N=2$, $y_1=x$, $y_2=y$, $p_1=y/2x$, and $p_2=1/2$ in \eqref{eq:lb3}, we obtain
\begin{equation*}
\Phi(y)\ge \frac{y}{2x}\Phi(x)+\frac{1}{2}\Phi(y)\iff \frac{\Phi(y)}{y}\ge \frac{\Phi(x)}{x}.
\end{equation*}
Interchanging the role of $x,y$, it follows that $\Phi(y)/y$ is constant, and hence $\Phi(y)=ky$ for some constant $k>0$ (because $\phi''>0$). Using the definition of $\Phi$ in \eqref{eq:defPhi} and letting $x=\phi^{-1}(y)$, we obtain
\begin{equation*}
\frac{\phi'(x)^2}{\phi''(x)}=k\phi(x)\iff \frac{\phi(x)\phi''(x)-\phi'(x)^2}{\phi'(x)^2}=a\quad \text{for $x\in I$},
\end{equation*}
where $a=1/k-1>-1$. Integrating both sides, we obtain
\begin{equation*}
-\frac{\phi(x)}{\phi'(x)}=ax+b\iff \frac{\phi'(x)}{\phi(x)}=-\frac{1}{ax+b},
\end{equation*}
where $a>-1$ and $b$ is such that $ax+b>0$ (because $\phi>0$ and $\phi'<0$). The forms of $\phi$ and $I$ in \eqref{eq:phiI} follow by integrating both sides with respect to $x$ and considering each case $-1<a<0$, $a=0$, and $a>0$ separately.

Conversely, if $\phi$ and $I$ take one of the forms in \eqref{eq:phiI}, then we have $\Phi(y)=\frac{y}{a+1}$ in \eqref{eq:defPhi}, and the inequality \eqref{eq:lb3} is trivial. Then we can go back the argument to show that $g$ is concave on its domain.
\end{proof}



\begin{thebibliography}{10}
\providecommand{\natexlab}[1]{#1}
\providecommand{\url}[1]{\texttt{#1}}
\expandafter\ifx\csname urlstyle\endcsname\relax
  \providecommand{\doi}[1]{doi: #1}\else
  \providecommand{\doi}{doi: \begingroup \urlstyle{rm}\Url}\fi

\bibitem[Cao and Werning(2018)]{CaoWerning2018}
Dan Cao and Iv{\'{a}}n Werning.
\newblock Saving and dissaving with hyperbolic discounting.
\newblock \emph{Econometrica}, 86\penalty0 (3):\penalty0 805--857, May 2018.
\newblock \doi{10.3982/ecta15112}.

\bibitem[Carroll and Kimball(1996)]{CarrollKimball1996}
Christopher~D. Carroll and Miles~S. Kimball.
\newblock On the concavity of the consumption function.
\newblock \emph{Econometrica}, 64\penalty0 (4):\penalty0 981--992, July 1996.
\newblock \doi{10.2307/2171853}.

\bibitem[Carroll and Kimball(2001)]{CarrollKimball2001}
Christopher~D. Carroll and Miles~S. Kimball.
\newblock Liquidity constraints and precautionary saving.
\newblock {NBER} Working Paper 8496, 2001.
\newblock URL \url{https://www.nber.org/papers/w8496}.

\bibitem[Gong et~al.(2012)Gong, Zhong, and Zou]{GongZhongZou2012}
Liutang Gong, Ruquan Zhong, and Heng-fu Zou.
\newblock On the concavity of the consumption function with the time varying
  discount rate.
\newblock \emph{Economics Letters}, 117\penalty0 (1):\penalty0 99--101, October
  2012.
\newblock \doi{10.1016/j.econlet.2012.04.086}.

\bibitem[Hardy et~al.(1952)Hardy, Littlewood, and
  P\'olya]{HardyLittlewoodPolyaInequalities}
Godfrey~H. Hardy, John~E. Littlewood, and George P\'olya.
\newblock \emph{Inequalities}.
\newblock Cambridge University Press, second edition, 1952.

\bibitem[Keynes(1936)]{keynes1936}
John~Maynard Keynes.
\newblock \emph{The General Theory of Employment, Interest, and Money}.
\newblock Harcourt, Brace \& World, New York, 1936.

\bibitem[Ma et~al.(2020)Ma, Stachurski, and Toda]{MaStachurskiToda2020JET}
Qingyin Ma, John Stachurski, and Alexis~Akira Toda.
\newblock The income fluctuation problem and the evolution of wealth.
\newblock \emph{Journal of Economic Theory}, 187:\penalty0 105003, May 2020.
\newblock \doi{10.1016/j.jet.2020.105003}.

\bibitem[Morris and Postlewaite(2020)]{MorrisPostlewaite2020}
Stephen Morris and Andrew Postlewaite.
\newblock Observational implications of non-exponential discounting.
\newblock \emph{Revue {\'E}conomique}, 71\penalty0 (2):\penalty0 313--322,
  March 2020.

\bibitem[Nishiyama and Kato(2012)]{NishiyamaKato2012}
Shin-Ichi Nishiyama and Ryo Kato.
\newblock On the concavity of the consumption function with a quadratic utility
  under liquidity constraints.
\newblock \emph{Theoretical Economics Letters}, 2\penalty0 (5):\penalty0
  566--569, 2012.
\newblock \doi{10.4236/tel.2012.25104}.

\bibitem[Stachurski and Toda(2019)]{StachurskiToda2019JET}
John Stachurski and Alexis~Akira Toda.
\newblock An impossibility theorem for wealth in heterogeneous-agent models
  with limited heterogeneity.
\newblock \emph{Journal of Economic Theory}, 182:\penalty0 1--24, July 2019.
\newblock \doi{10.1016/j.jet.2019.04.001}.

\end{thebibliography}

\end{document}